\documentclass[11pt]{article}
\usepackage{color}
\usepackage{enumerate}
\usepackage{citesort}
\usepackage{algorithm,algorithmic}
\usepackage{graphicx}

\usepackage{amsmath,amssymb}

\newtheorem{thm}{Theorem}
\newtheorem{lem}[thm]{Lemma}
\newtheorem{prop}[thm]{Proposition}
\newtheorem{defin}[thm]{Definition}
\newtheorem{coro}[thm]{Corollary}

\newenvironment{proof}[1][Proof]{\begin{trivlist}
\item[\hskip \labelsep {\bfseries #1}]}{\end{trivlist}}

\begin{document}

\title{Computing Minimum Tile Sets to Self-Assemble Colors Patterns\thanks{This work is supported in part by NSF Grants CCF-1049899 and CCF-1217770 to A.~C.~J.~and M-Y.~K. and HIIT Pump Priming Grants No. 902184/T30606 to S.~S.  It was published at ISAAC 2013, Hong Kong.}}

\author{
Aleck C. Johnsen\thanks{Department of Electrical Engineering and Computer Science, 
Northwestern University, Evanston, Illinois, USA. 
aleck.johnsen@gmail.com, kao@northwestern.edu}
\and 
Ming-Yang Kao\footnotemark[2]
\and 
Shinnosuke Seki\thanks{
Helsinki Institute of Information Technology (HIIT)}
\thanks{
Department of Information and Computer Science, Aalto University, P.~O.~Box~15400, FI-00076, Aalto, Finland. shinnosuke.seki@aalto.fi
}
}


\maketitle

\begin{abstract}
	Patterned self-assembly tile set synthesis (PATS) aims at finding a minimum tile set to uniquely self-assemble a given rectangular color pattern. 
	For $k \ge 1$, $k$-PATS is a variant of PATS that restricts input patterns to those with at most $k$ colors. 
	We prove the {\bf NP}-hardness of 29-PATS, where the best known is that of 60-PATS. 
\end{abstract}

	\section{Introduction}
Self-assembly, deterministic or otherwise, is a powerful force in nature.  Cosmological, biological, and even human sociological systems are all capable of self-assembling with a surprising amount of order on a macro scale, even as the processes used for their construction rely solely on local rules and effects.

We are interested in DNA computation and self-assembly and the tile assembly systems used to model it.  Winfree showed that such DNA computation is Turing-universal \cite{Win1}.  A natural question to then ask is, "what is a minimal tile set necessary for performing a certain computation?"  Barish et al. \cite{Bar} have early results for a binary counter.  Adleman was the first to also consider the time it takes for self-assembly, and also to ask a related question about minimal tile sets for particular patterns \cite{Adle}.

On self-assembly, Brun has introduced a series of tiling models to perform simple computations like multiplication\cite{B1}, and to non-deterministically solve NP-Hard problems \cite{B2,B3,B4}.  Two main goals are ``shape-building" and ``pattern-painting."  Real world molecules can potentially be designed to deterministically form a specific target shape, as implemented as DNA double-crossover molecules by Winfree \cite{Win3}.

This paper focuses on the self assembly of target color patterns rather than shapes.  Various patterns of practical significance were exhibited to self-assemble from DNA tiles, including Sierpinski triangle \cite{RPW} and binary counter \cite{Bar}.  As illustrated in Fig.~\ref{fig:binary_counter}, blue tile types B1, B2 and red tile types R1, R2 implement a half-subtractor as tile copies self-assemble the binary counter. 

Based on the abstract Tile Assembly Model (Winfree again) \cite{Win2}, Ma and Lombardy introduce{\em Patterned self-assembly tile set synthesis} (PATS) \cite{Ma1}.  PATS aims at finding a minimum tile type set with which a rectilinear tile assembly system (RTAS) uniquely self-assembles a given pattern.  {PATS} tries to compute the tile set of minimum size that uniquely self-assembles a given rectangular pattern.  The general case, allowing an arbitrary number of colors for the input pattern, was shown to be NP-hard \cite{Cze}.

For $k \ge 1$, {\em $k$-PATS} is a PATS' subproblem of practical significance in which inputs are restricted to patterns with at most $k$ colors.  Seki proved that 60-PATS is {\bf NP}-hard \cite{Seki}.  He designed a set $T$ of 84 tile types such it self-assembles a 3SAT evaluator circuit pattern $P(\phi)$, which is reduced from a given 3SAT instance $\phi$, if and only if the instance is satisfiable.  Gadget (sub-)patterns of $P(\phi)$ imply the property that $T$ is the unique smallest tile type set for $P(\phi)$ no matter what $\phi$ is.  Tile types in $T$ implement {\tt OR}, {\tt AND}, and {\tt NOT} gates in order to evaluate $\phi$.

Instead of 3SAT, in this paper we employ SUBSET SUM problem, which can be evaluated using just a half-subtractor similar to Brun \cite{B2}, and prove the following stronger result. 
\begin{thm}
\label{thm:main}
	29-PATS is {\bf NP}-hard. 
\end{thm}
\begin{figure}[tb]
\begin{center}
\includegraphics[width=160pt]{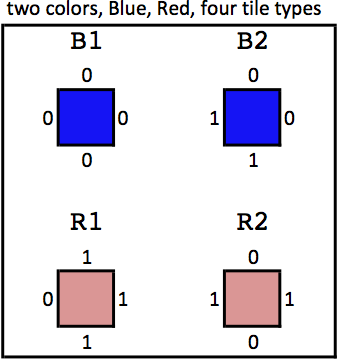}
\hspace{11pt}
\includegraphics[width=290pt]{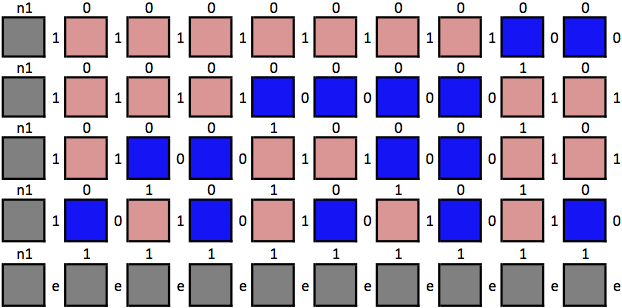}
\end{center}
\caption{We use four tile types, two colored Red and Blue each, to implement a half-subtractor.  For each tile type with glues named (n,s,e,w), we see that its glues perform the operation ``west minus south," then pass the output to the east and the carryout to the north.  Interpreting Red types as 1 bits and Blue types as 0 bits, with MSBs at the top, we see that the seed's north glues in row 0 interact with this tile type set to successively subtract 1 from one column to the next.}
\label{fig:binary_counter}
\end{figure}
\section{The Model and Preliminaries}
We begin our sufficient description of Rectilinear Tile Assembly System (RTAS) with a proxy for DNA proteins and their ``roles," and a method of ``attachment."  For a formal and excellent introduction, see \cite{Roth}.

A \textbf{(tile) type} is a unit-square with five parameters: four glues $g_{n},g_{s},g_{e},g_{w}$ associated with each edge north, south, east, and west; and a color $c$.  A \textbf{tile} is an instance of a tile type.  We assume tiles cannot be rotated or reflected.  By inclusion of a particular tile type in a set, we assume that we have access to infinitely many tiles of the type.  A \textbf{color} is an attribute of a type that can be shared by multiple types within a set- from a certain DNA perspective, it makes tile types indistinguishable and ``serve the same role."  A \textbf{glue} is an attribute of a type that is possibly equal to the glues of other types in a set, in which case the glues \textbf{match}- glues will tell us how we can ``stick together" tiles.

A \textbf{(rectangular) assembly} (of width $w$ and height $h$) is a function from a ``rectangular" subset $\left\{1,...,w\right\}\times \left\{1,...,h\right\}\subset \mathbb{Z}_+^2$ to types such that for any pair of adjacent inputs that share an ``edge," the output types' glues on the shared edge are the same.  A \textbf{(rectangular) color pattern} (of width $w$ and height $h$) is a function from a rectangle $\left\{1,...,w\right\}\times \left\{1,...,h\right\}\subset \mathbb{Z}_+^2$ to colors.

Our model for DNA self-assembly does not take place in a vacuum- we assume the preexistence of an L-shaped \textbf{seed}.  Our assembly function maps indexes in the 0-row and 0-column to seed types, which then have exposed north glues and east glues respectively.  The seed tiles are colored Slate Gray in all of our graphics.  The natural way to understand an assembly is to start with a seed, then the glues work to allow tiles to attach as follows.

A tile can \textbf{attach} at an index $(i,j)$ if and only if both its west glue matches the east glue of a tile already \textbf{placed} at $(i-1,j)$, and its south glue matches the north glue of a tile placed at $(i,j-1)$.
Therefore an assembly builds up from the south-west seed to the north-east.  If an assembly reaches a state such that no more tiles can attach, then it is a \textbf{terminal assembly}.

A \textbf{Rectilinear TAS} is a pair $\mathcal{T} = (T,\sigma_L)$ consisting of a tile type set $T$ and an L-shaped seed $\sigma_L$.  Its \textbf{size} is the cardinality of $T$.  We will use and refer to the requirement that every two tile types must have distinct (south,west) glue-tuples as \textbf{uniqueness}.  Two types that violate uniqueness are said to \textbf{clash}.

For RTASs satisfying uniqueness, we see by induction that fixing the exposed seed glues means that a tile placed at an empty index has necessarily deterministic type, and leads to a deterministic terminal assembly.  As we care mostly for colors, we will also say it \textbf{uniquely assembles a color pattern}.

As illustrated in Fig.~\ref{fig:binary_counter}, we see that \texttt{B2} uniquely attaches at index $(1,1)$ given the exposed seed glues.  The three tiles above it are successive copies of tile type \texttt{R2} and, attaching in turn, none could be replaced by any other tile type.  The assembly as shown is terminal.
The \textbf{Pattern self-assembly tile set synthesis (PATS)} problem \cite{Ma1,Ma2} asks the question, ``given a rectangular color pattern $\mathcal{P}$, what is a RTAS of minimum size that uniquely \textbf{solves} the pattern?"
\begin{defin}
{\scshape $k$-colored Pats ($k$-Pats)} 
\ \\
\begin{tabular}{ll}
{\scshape Input}: & a $k$-colored pattern $\mathcal{P}$; \\
{\scshape Output}:  & a smallest RTAS which uniquely self-assembles $\mathcal{P}$. 
\end{tabular}
\end{defin}
60-PATS was shown to be NP-Hard by Seki \cite{Seki}.  We endeavor to show that 29-PATS is NP-Hard by a polynomial reduction from Subset Sum\footnote{We use the variant of Subset Sum that restricts all elements of a set $S$ to be positive integers, and asks if any subset $S^*\subseteq S$ sums to some target $n\in \mathbb{N}$.} to the size variant of 29-PATS: given a pattern in 29 colors, what is the minimum size of any RTAS that uniquely self-assembles the given pattern.

For the rest of this paper, we assume we have a black box capable of solving the size variant of 29-PATS.
\section{The Circuit}
We reduce an instance $\textbf{SS} = (S,n)$ of Subset Sum to a pattern in 29-colors \texttt{PATTERN} such that
\begin{enumerate}
\item there exists a RTAS $(T,\sigma^*)$ with $|T|=46$ that uniquely assembles \texttt{PATTERN} if and only if \textbf{SS} is solvable;
\item there does not exist any tile type set $T'\neq T$ (ignoring isomorphisms) of size 46 or less such that a RTAS $(T',\sigma')$ can uniquely self-assemble \texttt{PATTERN}.
\end{enumerate}
In this section, we verify (1) for a sub-pattern of \texttt{PATTERN} called \texttt{CIRCUIT} using 26 of the 46 tile types of $T$.\footnote{The other 20 tile types will be necessary and sufficient regardless of the instance \textbf{SS}.}  Specifically, we introduce a subset of the critical tile type set $T$ in Fig.~\ref{fig:tileset1}.  An example of the reduction to \texttt{CIRCUIT} is given in Fig.~\ref{fig:brunexample}.  For solvable \textbf{SS}, $T$ can uniquely assemble \texttt{CIRCUIT}, i.e., $T$ is \textit{sufficient}; therefore, $|T|$ is an upper bound on the output of the black box in these cases.  The reduction works as follows:\\
\begin{figure}[tb]
\begin{center}
\includegraphics[width=450pt]{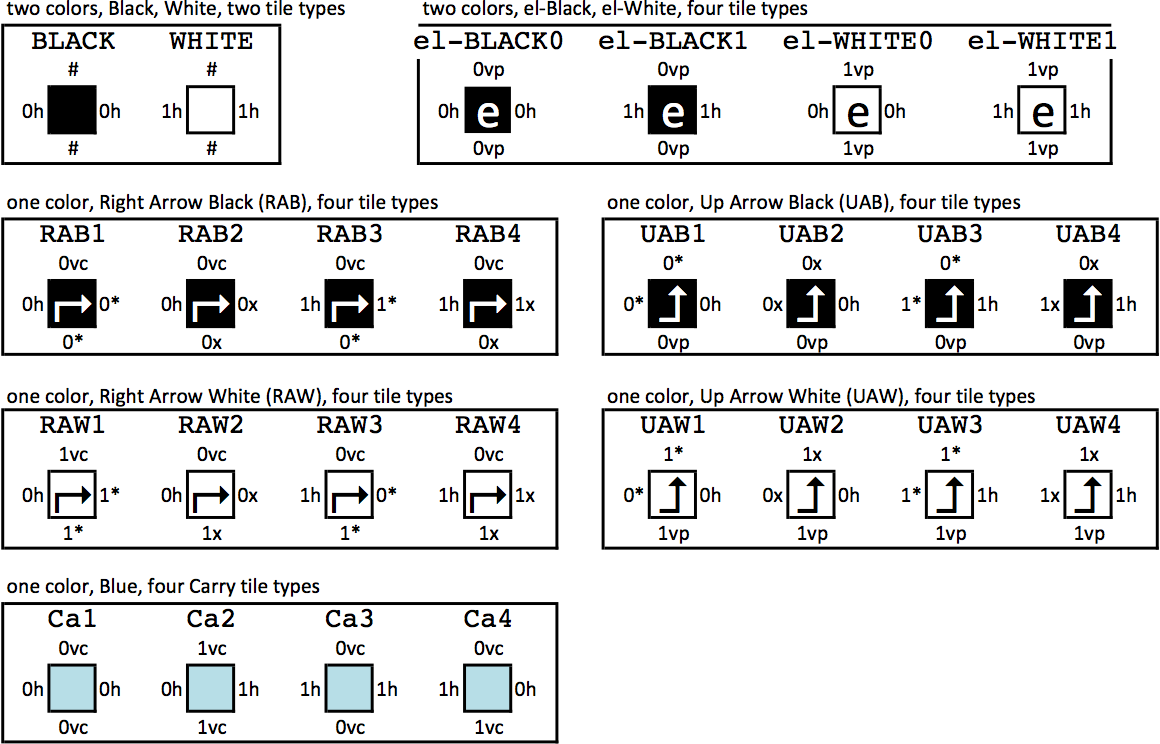}
\end{center}
\caption{These are the 26 tiles types in 9 colors that currently make up the critical set $T$, whether an $S^*$ solving \textbf{SS} exists or not.}
\label{fig:tileset1}
\end{figure}
We determine the height of the circuit to be $\max(\lceil\log_2(n+\Sigma_{\text{all}})\rceil,21)$, where $\Sigma_{\text{all}}$ is the sum of all elements in $S$.\footnote{Later we add a sub-pattern which requires height of at least 21.}  Then $2^{\text{row\#}}$ is greater than both $n$ and $\Sigma_{\text{all}}$. The width of \texttt{CIRCUIT} will be determined online.  Then:
\begin{figure}[tb]
\begin{center}
\includegraphics[width=403pt]{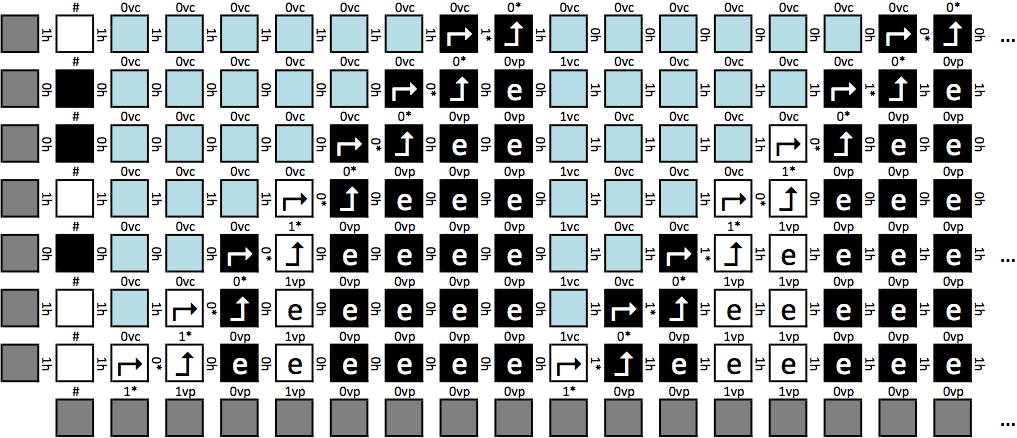}
\includegraphics[width=403pt]{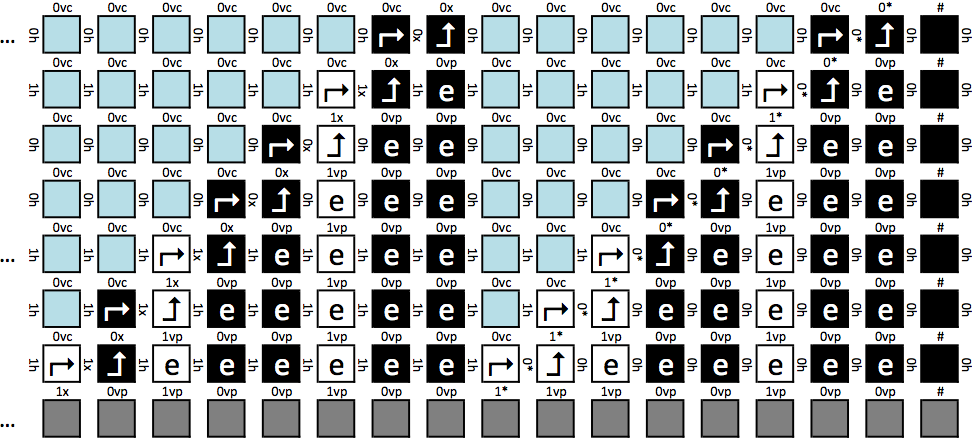}
\end{center}
\caption{For space considerations, and WLOG here, we draw this example with height of 7 rather than 21.  $\textbf{SS} = (S = \left\{11,25,37,39\right\}, n = 75)$, the same example as Brun in \cite{B2} for easy comparison of the tile sets.  $S^* = \left\{11,25,39\right\}$ so we see the seed north glue of 1* in columns 2, 10 and 26; and seed north glue 1x in column 18.}
\label{fig:brunexample}
\end{figure}
\begin{itemize}
\item The target number $n$ from \textbf{SS} is encoded in binary- MSB at the top- as colors of tiles in the first column, Black for 0 and White for 1.  Unique tile types of these colors in $T$ imply that equivalently, $n$ is also encoded in the east glues of this column.
\item Using $(1+\text{row\#})$ columns each, we similarly encode elements of $S$ in succession, MSBs on the right.  A 0 bit is encoded by a column of colors Right-Arrow Black (RAB), Up-Arrow Black (UAB), and el-Black; a 1 bit by RAW, UAW, and el-White.
\item Given our tile type set $T$ and reduction, there are only two colors appearing next to seed tiles in row 1 or column 1 such that the seed has a choice in glues: tile types of colors RAB and RAW have varying south glues.  Compare this to the south glues of UAB for example, whose tile types all have south glue 0vp.
\item Then given $T$, the set of feasible seeds can only vary in their choice of glues in the first column of the encoding of each element of $S$; the choice is to ``tag" the glue `*' to mean we should subtract the current element, or `x' to mean we should not.
\item The Right-Arrow (RA) and Up-Arrow (UA) tile types cooperate to transfer the subtraction ON/OFF signal through an element's columns until it runs off the top.
\item Subtraction of a bit takes place in the RA squares.  Tile types below them preserve bit information of the element vertically.  If subtraction is ON, the RA tile type works as a half-subtractor (like in Fig.~\ref{fig:binary_counter}): it receives a ``running total" bit from the west, subtracts from it the south bit, outputs to the east and ``carries" to the north.  If subtraction is OFF, the RA tile type preserves bit information horizontally (passes it through) and carries 0.
\item All indexes above the squares colored RAB and RAW are required to be color Carry Blue.  These tile types are half-subtractors again, to complete the carry operation correctly from column to column.
\item Below the RA colors, subtraction has already taken place so bit information is preserved horizontally for use by the next element of $S$ to the right.
\item The last column of \texttt{CIRCUIT} is all Black.  We started with $t$, we subtracted some subset of elements of $S$.  For an instance of \textbf{SS} for which $S^*$ exists, we should finish with all-Black 0.
\end{itemize}
\begin{prop}
\label{thm:section3main}
Given an \textbf{SS} and our tile type set $T$, the RTAS $(T,\sigma^*)$ with $|T| = 26$ uniquely assembles \texttt{CIRCUIT} if and only if \textbf{SS} is solvable.
\end{prop}
\begin{proof}
The reduction correctly encodes $n$ and the elements of $S$, and correctly does subtraction or not.  The proof is in the previous analysis, for further justification see Brun \cite{B2}.\\
Then our choice over feasible seed glues is exactly equivalent to choices over subsets $S'\subseteq S$ of elements to subtract from $n$ to try to obtain 0.  $\bullet$
\end{proof}
\section{Minimum Tile Complexity}
In the last section, we described a tile type set $T$ and showed that it is \textit{sufficient} to self-assemble the color pattern \texttt{CIRCUIT} that results from reducing any satisfiable instance of Subset Sum, given an appropriate seed $\sigma^*$.

At this point, the size of $T$ is an upper bound on what the black box can return on a satisfiable input \texttt{CIRCUIT}.  But in a trivial example in which the elements of $S$ sum to $n$ exactly, we have no use for any tile types utilizing the glues 0x or 1x so output will be less.  Worse, it is possible that given an unsolvable instance, the black box outputs 26 or less using some other $T'$, making it indistinguishable from a solvable instance of subset sum.

So we have a need to establish a lower bound on the output from the black box for any input.  To do this, we add the tile types in Fig.~\ref{fig:tileset2} to $T$.

\begin{figure}[tb]
\begin{center}
\includegraphics[width=420pt]{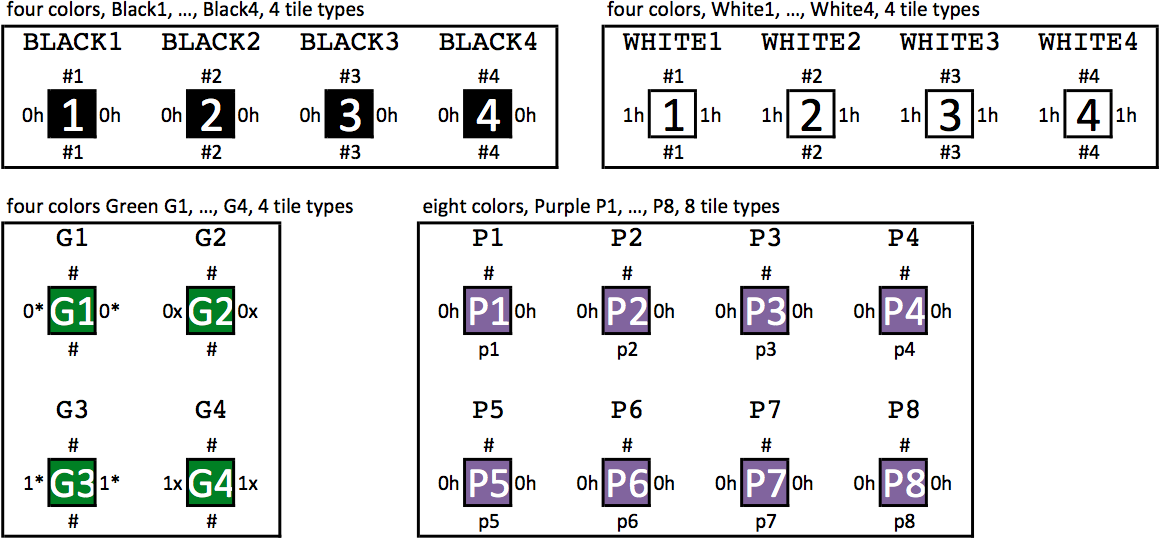}
\end{center}
\caption{We add these 20 tile types to our set $T$.  While they are not necessary to uniquely assemble \texttt{CIRCUIT}, adding them does not affect the sufficiency of $T$, now of size 46; they will be used in auxiliary patterns that will be attached to \texttt{CIRCUIT} such that all 46 tiles are necessary.}
\label{fig:tileset2}
\end{figure}
We will join \texttt{CIRCUIT} with auxiliary patterns into the full input \texttt{PATTERN}; then for each color $c$ that appears in \texttt{PATTERN}, we will prove that the number of tile types colored $c$ needed by any set $T'$ that uniquely assembles \texttt{PATTERN} must be at least the number used by $T$ of color $c$.
Finally, summing over types will give us our lower bound- the black box cannot output less than 46.\\
Auxiliary patterns must meet the following constraints: they must assemble using $T$, i.e., be \textbf{consistent}; they must also ``splice" together using $T$.  Failing either, we lose the sufficiency of $T$.  However, these constraints do endow a benefit: because $T$ is sufficient to assemble \texttt{CIRCUIT} and all auxiliary patterns for reductions from satisfiable instances of Subset Sum, we know for these cases that the black box need not consider tile type sets larger than 46.

So the goal of this section and the next is to design auxiliary sub-patterns such that, including them in \texttt{PATTERN}, the tile set $T$ will be \textit{necessary}.
Sub-patterns in this section are named \texttt{LB\#} for Lower Bound.  We claim that $T$ is sufficient for all auxiliary patterns in Sections 4 and 5 without proof.  The strategy for joining together the various sub-patterns is given as part of correctness in Section 6.

We take it as a given that the 29 colors described in $T$ will be used in some sub-pattern of \texttt{PATTERN}.
There are 22 colors in \texttt{PATTERN} with just one tile type of their color in $T$.\footnote{Colors Black and White from Fig.~\ref{fig:tileset1}, and all 20 colors from Fig.~\ref{fig:tileset2}.}  Each trivially requires one type, so we have finished 22 of our 29 necessary lower bound proofs, leaving 17 tiles unassigned, or \textbf{free}.

Sub-pattern \texttt{LB1} is given in Fig.~\ref{fig:lb1} and is used in the next result.
\begin{figure}[tb]
\begin{center}
\includegraphics[width=420pt]{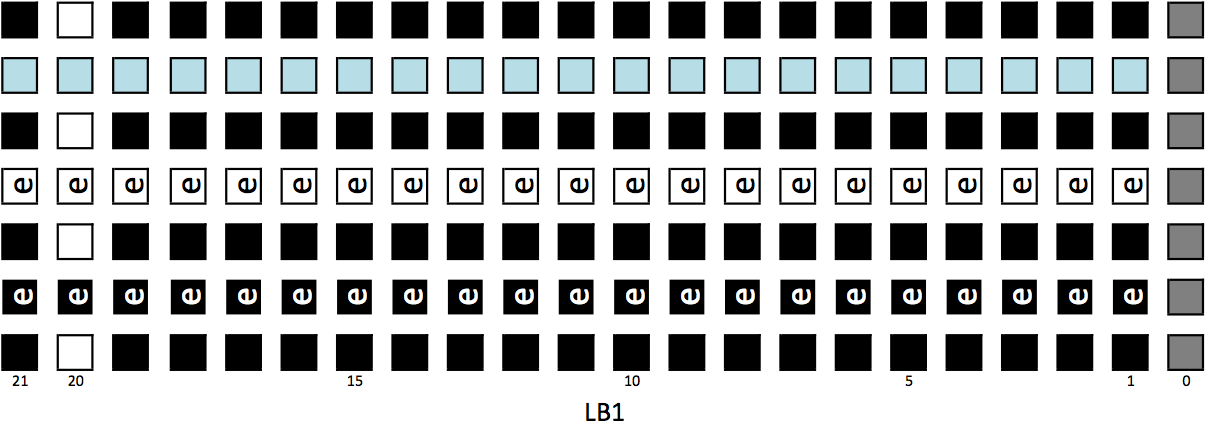}
\end{center}
\caption{\texttt{LB1} is presented sideways for space considerations.  The Slate seed tiles should be interpreted as being the 0-row, then it is the rows that are numbered up to its height of 21.  \texttt{LB1} guarantees that there are at least 2 tile types colored each of el-Black, el-White, and Carry Blue.}
\label{fig:lb1}
\end{figure}
\begin{prop}
\label{thm:lb1}
If \texttt{LB1} appears as a sub-pattern of \texttt{PATTERN}, then in any minimum tile type set solving \texttt{PATTERN}, either there must be at least 2 tile types with distinct east glues of each color el-Black, el-White, and Carry Blue as we suggest in $T$; or the black box must find a minimum tile type set larger than 46.
\end{prop}
\begin{proof}
By symmetry, we only need to prove the case for el-Black.  We assume that the black box can find a solution using at most 46 tile types.\\
By contradiction, assume there is only one tile type el-Black.  Then there is only one glue east of column 2, call it e.  The Black types can have at most 18 unique north glues in column 3, so by the 19th north glue, they must cycle (because the west glue is a constant).  But this requires a Black tile at $(3,20)$, not White.  $\bullet$
\end{proof}
Proposition~\ref{thm:lb1} gave us 2 tile types for Blue Carry tiles, but we need a lower bound of 4.  We will obtain it in Proposition~\ref{thm:lb43prop}, first we need the next lemmas.

Proposition~\ref{thm:lb1} also assigned 3 of our free tile types.  Now of 46 types from our upper bound, we have assigned 32, leaving 14 free.

We start by generalizing the result in Proposition~\ref{thm:lb1}, using \texttt{LB0} in Lemma~\ref{thm:lb0}.
\begin{figure}[tb]
\begin{center}
\includegraphics[width=240pt]{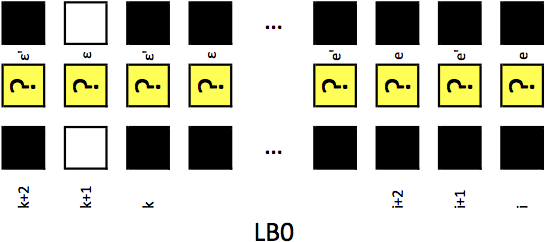}
\end{center}
\caption{\texttt{LB0} is presented sideways to conserve space.  Looking at it from the side, variables on the left are row indexes, with $i$ at the bottom.  There are two alternating glue names between the second and third columns, e and e$'$.  Above the vertical ellipsis, they are represented by $\epsilon$ and $\epsilon'$, such that they still alternate, but whether e is equal to $\epsilon$ or $\epsilon'$ depends on the parity of the number of rows.  It is also possible that there is only one glue. (Then e is equal to e$'$ is equal to $\epsilon$ is equal to $\epsilon'$.)  In either case, given enough consecutive rows of Black tiles, new Black types must eventually be exhausted and the vertical glues in the third column must ``cycle."  Then the White tile at $(3,k+1)$ must clash with a Black type, so there must actually be at least two types of the variable Yellow color with distinct east glues.  As a final note, the seed is not necessarily present- we can start in row $i$.}
\label{fig:lb0}
\end{figure}
\begin{lem}
\label{thm:lb0}
\textit{If \texttt{LB0} appears as a sub-pattern of \texttt{PATTERN}, such that starting in row $i$, there are $(k-i+1)$ consecutive rows of Black-Yellow-Black as shown such that the east glues of the second column are either all the same or alternate between two glues, and the row $(k+1)$ continues the glue pattern but requires White tiles in the first and third columns, and such that $(k-i)$ is at least the maximum number of Black types, then in any minimum tile type set solving \texttt{PATTERN}, either the Yellow tile in row $(k+1)$ must have an east glue distinct from e and e$'$, and therefore be a different type than any appearing below it; or the black box must find a minimum tile type set larger than 46.}
\end{lem}
\begin{proof}
We assume that the black box can find a solution using at most 46 tile types.  Then by contradiction, assume that we can place a Yellow tile at $(2,k+1)$ that has east glue $\epsilon$.  In two cases. First, assume the east glue of every tile in column 2 is known to be the same, e.

Assume $M$ is the Minimum Tile Type Set found by the black box.  Now consider the subset $Be\subset M$ to be the subset restricted to only types colored Black with west glues e.  These are the candidate types whose tiles can be placed in the third column.  Reduce $Be$ to a directed graph $G$ where each type $x$ is represented by a node with out-degree at most 1: each node $v_x$ points to the node of the type in $Be$ whose south glue is equal to $x$'s north glue, if such a type exists.\footnote{Common west glues e, and then uniqueness of the tile type set, guarantees that at most one other type has the critical south glue, hence, out-degree at most 1.}

Now for any tile type $y\in Be$ placed at the ``bottom" index $(3,i)$ of \texttt{LB1}, it must be that the type placed directly north of it can be determined by following the edge out of $y$'s node in $G$, as this neighbor of $y$ must have west glue e (by its inclusion in $Be$), and south glue equal to $y$'s north glue (by design of the edge).

This logic repeats as we move north placing tile types in the third column in \texttt{LB1}, and we see that the order of types placed is determined by following out-edges in $G$.

However, the black box can assign the color Black to at most $(k-i)$ types.  Then we know that for any starting node in $G$ (and corresponding type placed at $(3,i)$), after $(k-i)$ hops we must be in some cycle of $G$.  Why?  Because there can be at most $(k-i)$ unique north glues, therefore the $(k-i+1)$th north glue must repeat.  This cycle corresponds to a cycle of vertical glue names.\footnote{The path can not end prematurely because, by the definition of $M$, it solves the color pattern.}  Further, as we continue to place types moving north, uniqueness requires that we stay in this cycle as long as the west glues of the types must be e.  But then the type at $(3,k)$ must be Black, not White, giving us our contradiction.

In the case that there are two alternating east glues in column 2, we call them e and e$'$, and the proof proceeds much the same.

$Be'$ now contains Black types with west glue either e or e$'$.  The reduction is to a bipartite graph $BG'$ (one partition for types with each west glue), where each type $x'$ has a node $v_x'$ with out-degree at most 1, and points to the node of the type $y'$ such that $y'$'s west glue is different than $x'$'s (because the glues alternate) and $y'$'s south glue is equal to $x'$'s north glue.  Then again edges represent ``stack-ability" in column 3.  And again, there are at most $(k-i)$ nodes, so after starting by placing any type in $(3,i)$ and making $(k-i)$ hops, we are in a cycle.  This cycle guarantees that the north glue at index $(3,k)$, combined with glue $\epsilon$, is known to require placing a Black type at $(3,k+1)$.  This contradiction completes the second case.  $\bullet$
\end{proof}
We return to proving our specific lower bounds.  Recall \texttt{LB1} in Fig.~\ref{fig:lb1} on page~\pageref{fig:lb1}.
\begin{lem}
\label{thm:twoselfstack}
If \texttt{LB1} appears as a sub-pattern of \texttt{PATTERN}, then in any minimum tile type set solving \texttt{PATTERN}, either it is true that if there are exactly 2 Blue Carry tile types, then the north and south glues of both types must all be the same glue; or the black box must find a minimum tile type set larger than 46.
\end{lem}
\begin{proof}
We assume that the black box can find a solution using at most 46 tile types.  From Proposition~\ref{thm:lb1} we can assume that the 2 Carry types have distinct east glues.  All cases use \texttt{LB1}.\\
\textbf{Case 1.  Neither type self-stacks vertically.}\footnote{By ``self-stack vertically," we mean that a tile type has north glue equal to its south glue.}  In this case the types must alternate vertically- if they do not, then combined with neither type self-stacking, it will not be possible to place types in more than the first two rows of \texttt{LB1} column 6.  However if their north and south glues do allow them to alternate, then we can apply Lemma~\ref{thm:lb0} to see that we need a type at $(6,20)$ with an east glue unique from both types that we have, completing this case.\\
\textbf{Case 2a.  Exactly one type $x$ self-stacks vertically, and $x$ is placed at $(6,1)$ of \texttt{LB1}.}  In this case, Lemma~\ref{thm:lb0} can force our non-self-stacking type $y$, to be placed at $(6,20)$, if not at a smaller row-index in column 6.  This implies that $y$'s south glue equals the north-and-south glue of $x$.  But then neither type can place at index $(6,21)$ (or at some other index directly above $y$ if it is used earlier), as $x$ and $y$ have the same south glues, while $y$'s north glue is distinct.\\
\textbf{Case 2b.  Exactly one type $x$ self-stacks vertically, but $y$ is placed at $(6,1)$.}  In this case $y$ doesn't self-stack so we must place $x$ at $(6,2)$.  Indeed, we must place $x$ at every remaining index of column 6, as neither type can have south glue equal to $y$'s north glue.  Then there are 18 consecutive rows of Black-Blue-Black with one critical glue between the second and third columns and 14 free tiles, so by Lemma~\ref{thm:lb0}, the type at $(6,20)$ must have a different east glue than the consecutive types before it.  But there is only one other type we can place here, $y$, whose south glue does not match, giving us the contradiction.\\
\textbf{Case 3.  Both types self-stack, but do not stack with each other.}  In this case, regardless of which type we place at $(6,1)$, it can only stack with itself all the way up to row 19.  But again we can apply Lemma~\ref{thm:lb0} to force $(6,20)$ to be a distinct type, which is impossible because the second's south glue does not equal the north glue of the first.\\
The only remaining possibility is that both Carry Blue types self-stack, and stack with each other, which is our desired result.  $\bullet$
\end{proof}
\begin{figure}[tb]
\begin{center}
\includegraphics[width=330pt]{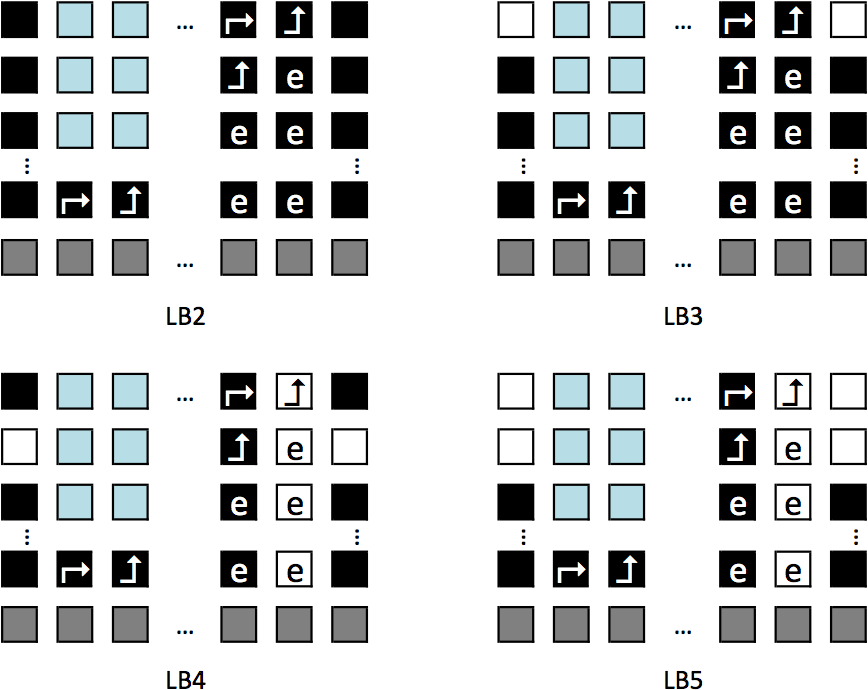}
\end{center}
\caption{It should be understood that the sub-patterns that appear above the ellipsis in \texttt{LB2} through \texttt{LB5} will be appear as the top rows when placed within \texttt{PATTERN}.  These color sub-patterns are the result of subtracting 0 (all element bits 0) in \texttt{LB2} and \texttt{LB3} and $2^{\text{row\#}}$ in \texttt{LB4} and \texttt{LB5} (all element bits 0 except the MSB is 1) from the ``numbers" given in the respective first columns.}
\label{fig:lb2}
\end{figure}
We make a slight detour from the Carry Blue lower bound and turn our attention to sub-patterns that include the UAB, UAW, RAB, and RAW tiles.  We prove preliminary results now, and return to some of these sub-patterns later.  Obtaining the results in this order allows us to ``save" colors.

We present sub-patterns \texttt{LB2} through \texttt{LB5} for Lemma~\ref{thm:lb2} in Fig.~\ref{fig:lb2}.
\begin{lem}
\label{thm:lb2}
If \texttt{LB2} through \texttt{LB5} appears as a sub-patterns of \texttt{PATTERN}, then in any minimum tile type set solving \texttt{PATTERN}, we must have at least 2 types colored UAB or at least 2 types colored Black; and at least 2 types colored UAW or at least 2 types colored White.
\end{lem}
\begin{proof}
By symmetry, we only prove the case for UAB and Black.  By contradiction, assume there is only one tile type colored UAB with east glue e, and one tile type colored Black with north glue n.  Then for both \texttt{LB2} and \texttt{LB3}, the tile placed at their respective $(6,4)$ indexes must have west glue e and south glue n.  But by uniqueness, it is impossible for two different tiles to do this, one colored Black as needed by \texttt{LB2} and one colored White as needed by \texttt{LB3}.  $\bullet$
\end{proof}
\begin{figure}[tb]
\begin{center}
\includegraphics[width=380pt]{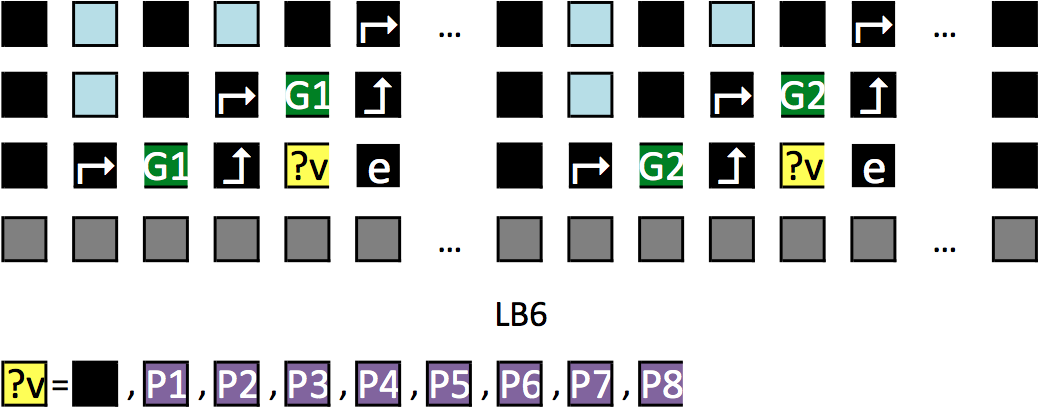}\\
\vspace{15pt}
\includegraphics[width=380pt]{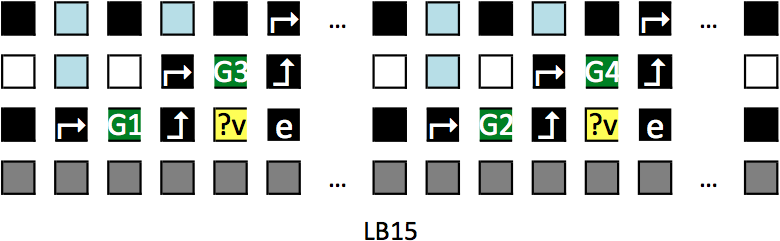}
\end{center}
\caption{There cannot be only one tile type colored RAB, because then the Green types would all need distinct south glues, and then there are not enough free tile types for Purple types to match each Green south glue.}
\label{fig:lb6}
\end{figure}
In Lemma~\ref{thm:lb6} we obtain a result for RAB and RAW similar to Lemma~\ref{thm:lb2}, again to save colors.  
We present sub-patterns \texttt{LB6} and \texttt{LB15} in Fig.~\ref{fig:lb6}, and \texttt{LB24}, and \texttt{LB33} in Fig.~\ref{fig:lb24}, for the results in Lemma~\ref{thm:lb6} and Lemma~\ref{thm:lb6b}.

It should be understood that there are 9 total copies of each sub-pattern, one for each color ``value" that the variable-Yellow indexes should take.  For example, the patterns \texttt{LB6} through \texttt{LB14} are given at the top of Fig.~\ref{fig:lb6} with variable-Yellow replaced by Black and each of P1 through P8 respectively.  These 36 sub-patterns will be used again in Proposition~\ref{thm:lb43} to obtain a lower bound of 4 for RAB and RAW.

Of 46 tile types, we have assigned 32; 1 more is known to be UAB or Black; another 1 more is known to be UAW or White; 12 are free to be any color.

\begin{figure}[tb]
\begin{center}
\includegraphics[width=380pt]{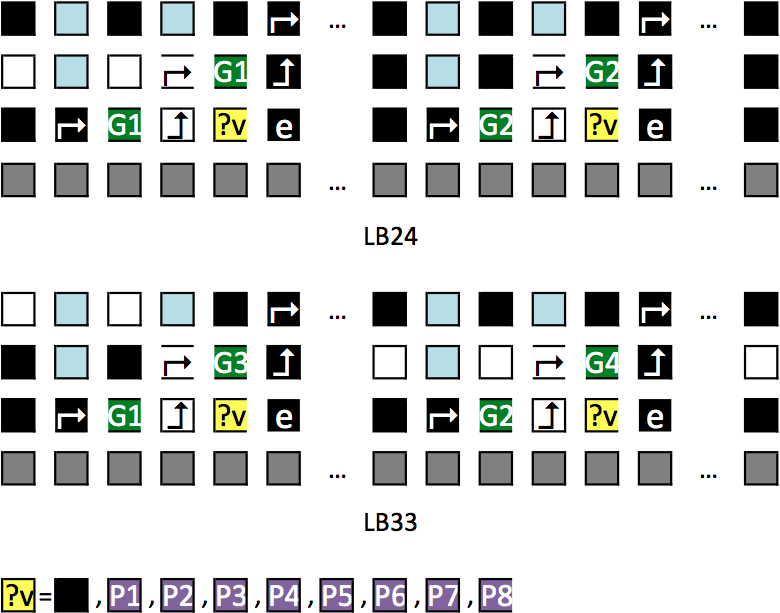}
\end{center}
\caption{\texttt{LB24} through \texttt{LB41} obtain the symmetric result for RAW as described for RAB in Lemma~\ref{thm:lb6}.  In comparison to Fig.~\ref{fig:lb6}, tiles in the first column and the output column have been changed to maintain consistency with $T$, reflecting the different nature of the RAB and RAW tiles.}
\label{fig:lb24}
\end{figure}
\begin{lem}
\label{thm:lb6}
If \texttt{LB1} through \texttt{LB41} appear as sub-patterns of \texttt{PATTERN}, then in any minimum tile type set solving \texttt{PATTERN}, either we must have at least 2 types colored RAB and at least 2 types colored RAW as we suggest in $T$; or the black box must find a minimum tile type set larger than 46.
\end{lem}
\begin{proof}
We prove the implied lower bound for RAB, then RAW by symmetry.  We assume the black box can find a solution using at most 46 tile types.\\
By contradiction, assume there is only one tile type colored RAB.  Then uniqueness requires that there must be four distinct tile types, one with each of the four Green colors, with distinct south glues.  Four distinct south glues at the indexes colored G1, ..., G4 in \texttt{LB6} through \texttt{LB23} would require 24 free tile types for the various Purple colors.  But we only have 12 types free to be any Purple color.  $\bullet$
\end{proof}
We also get the following result from \texttt{LB6} through \texttt{LB41}.
\begin{lem}
\label{thm:lb6b}
If \texttt{LB6} through \texttt{LB41} appear as sub-patterns of \texttt{PATTERN}, then in any minimum tile type set solving \texttt{PATTERN}, either it must be true that if the number of free tile types remaining is insufficient to allow for there to be two types of each of the 9 yellow-Variable colors, then there are at least 4 tile types with distinct east glues colored RAB and at least 4 tile types with distinct east glues colored RAW; or the black box must find a minimum tile type set larger than 46.
\end{lem}
\begin{proof}
We prove the implied lower bound for RAB, then RAW by symmetry.  We assume the black box can find a solution using at most 46 tile types.

From the Lemma statement, we can also assume that there are not enough free tile types to assign them colors Black and the Purple colors P1, ..., P8 such that each of these 9 colors has at least 2 tile types assigned to its color.  Therefore there must be a color in this set such that there is only one type of its color, call it $C$; and the unique type colored $C$ has north glue n.

Now we analyze in the two sub-patterns described in Fig.~\ref{fig:lb6} corresponding to variable-Yellow taking the color $C$.  We see that n must be the south glue of at least one type colored by each of the 4 Green colors.  By uniqueness, these 4 distinct types colored G1, G2, G3, and G4 must have distinct west glues given their common south glue n.  This would require at least 4 tile types colored RAB with distinct east glues.  $\bullet$
\end{proof}
We need sub-pattern \texttt{LB42} in Fig.~\ref{fig:lb42} and one more lemma before proving a minimum of 4 tile types colored Carry Blue.  Lemma~\ref{thm:lb42} is technical and three pages long, and can be skipped.
\begin{lem}
\label{thm:lb42}
If \texttt{LB42} appears as a sub-pattern of \texttt{PATTERN}, along with \texttt{LB1}, then in any minimum tile type set solving \texttt{PATTERN}, either it must be true that if there are less than 4 Blue Carry tile types, then there must be at least 4 tile types colored Black or White combined; or the black box must find a minimum tile set larger than 46.
\end{lem}
\begin{figure}
\begin{center}
\includegraphics[width=220pt]{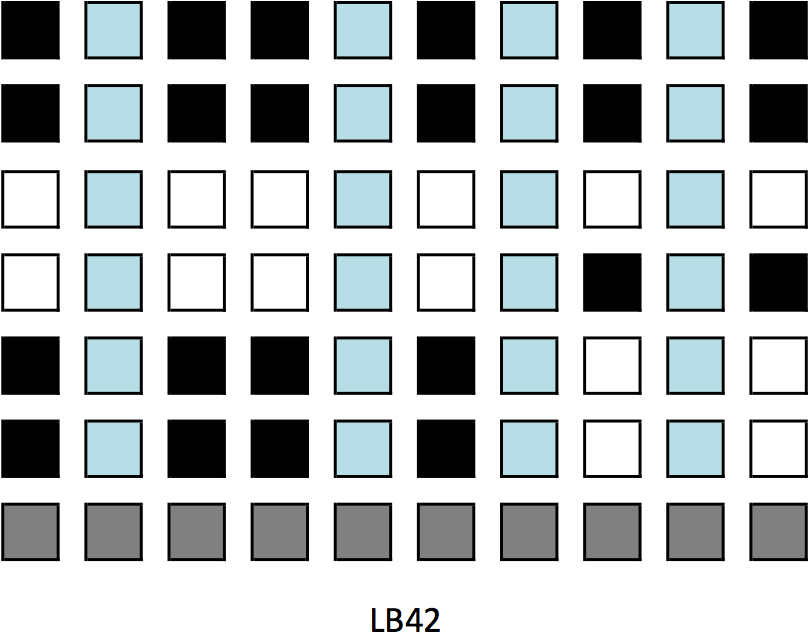}
\end{center}
\caption{Blue Carry tiles act as half-subtractors, an entire column of them can subtract 0 or 1.  If there are less than four Blue Carry tile types, then there must be at least four tile types of colors Black and White combined.}
\label{fig:lb42}
\end{figure}
\begin{proof}
We assume that the black box can find a solution using at most 46 tile types, then we do a case analysis.  From Proposition~\ref{thm:lb1} we know that there must be at least 2 Carry types.  And of course we must have at least 1 Black and 1 White type, 2 total.  So we consider the following six possible combinations and show them all to be impossible: there are either 2 or 3 Carry types; and there are either 1 Black and 1 White, or 1 Black and 2 White types, or 2 Black and 1 White.\\
\textbf{Cases 1\&2.  2 or 3 Carry types, 1 Black, 1 White.}  Because of the Black and White sub-patterns that appear in columns 3 and 4 of \texttt{LB42}, we know that the unique Black and White types can be said to have the glues we suggest in $T$.\footnote{There is only one of each, and both Black and White types self-stack vertically and self-stack horizontally; and also stack vertically with each other forcing equality of all vertical glues.}  Then the horizontal glues of the Black and White types are distinct, which means that we need at least 4 distinct Carry tile types to place at $(5,2)$, $(5,3)$, $(7,2)$, and $(7,3)$, because these indexes require distinct (east,west) glue pairs.

For the more complicated remaining cases, we use the following naming conventions for the types, to the extent that they exist in the Case: Black types are B1 and B2; White types are W1 and W2; Carry types are C1, C2 and C3.  When appropriate, as a convenience for both us and the reader, we will assign glue names to these types as to keep them consistent with $T$.\\
\textbf{Case 3.  2 Carry types, 2 Black, 1 White.}  There is only 1 White tile type W1, and it stacks with itself both vertically and horizontally, so WLOG we can fix its glue names as suggested in $T$.  So W1's (n,s,e,w) is (\#,\#,1h,1h).

Now we know that there must exist a Carry type with east and west glue both 1h from index $(5,3)$ of \texttt{LB42} between two White tiles, so we assign these glues to C2.

Also we see that the Black tile at $(8,3)$ has White tiles above and below it, so we can say the type placed here is B1 and has north and south glues \#.  Also it must now have a west glue unique from W1's, or it will violate uniqueness.  So we can call its west glue 0h.

Now the tile at $(7,3)$ has east glue 0h.  Type C2 has already been assigned east glue 1h, so we must place C1 at $(7,3)$.  We also see that from a White tile at $(6,3)$, this index, and therefore type C1, must have west glue 1h.

But now by Lemma~\ref{thm:twoselfstack}, we know that if there are only two Carry tile types, they must have all four of their north and south glues the same, call it g.  But now types C1 and C2 have the same (south,west) glue pair (g,1h), violating uniqueness.\\
\textbf{Case 4.  2 Carry types, 1 Black, 2 White.}  There is only 1 Black type B1, and it stacks with itself both vertically and horizontally, so we can fix its glue names as suggested in $T$.  So B1's (n,s,e,w) is (\#,\#,0h,0h).

Now we know that there must exist a Carry type with east and west glue both 0h from index $(5,1)$ of \texttt{LB42} between two Black tiles, so we assign these glues to type C1.

Also we see that the White tile at $(10,4)$ has Black tiles above and below it, so we can say the tile placed here is type W1 and has north and south glues \#.  Also it must now have a west glue unique from B1's, or it will violate uniqueness.  So we can call its west glue 1h.

Then the Carry tile at $(9,4)$ must have east glue 1h and be type C2, as C1 has already been assigned east glue 0h.

Now.  Neither the tile at $(8,1)$ nor the tile at $(8,2)$ can be type W1.  Otherwise $(7,1)$ or $(7,2)$ has east glue 1h, which makes one of them type C2, as C1 is known to have east glue 0h.  But then C2 has west glue 0h, and as in the last Case, by Lemma~\ref{thm:twoselfstack} both types must have south glue g, so C1 and C2 have identical (south,west) glue pairs, (g,0h), violating uniqueness.

So instead tiles at both $(8,1)$ and $(8,2)$ must be type W2.  Then W2 both self-stacks vertically and has north glue at $(8,2)$ equal to B1's south glue \#, so W2's south glue must be \# as well.  We must place C1 known to have east glue 0h, or C2 known to have east glue 1h, at $(7,1)$, because we only have two Carry tiles.  But either east glue when it becomes the west glue of W2, will cause W2 to clash with either the (south,west) glue pair of B1 (\#,0h), or W1 (\#,1h), giving us our contradiction.\\
\textbf{Case 5.  3 Carry types, 1 Black, 2 White.}  We start the same as case 4.  There is 1 Black type B1 with glues (n,s,e,w) equal to (\#,\#,0h,0h).  WLOG we will say that a type C1 is placed at $(5,1)$ and has east and west glues 0h.  From $(10,4)$, we can say that type W1 has north and south glue \#, and then must have west glue distinct from 0h, so call it 1h.  The tile at $(9,4)$ has east glue 1h, so we can say that it is C2.

Now we try to rule out the possibility that C3 has east glue 0h or 1h.  By contradiction, assume that it does.  Then every Carry type has east glue 0h or 1h.  As a result, W2 must have a south glue distinct from \#.  To see this, we realize that W2 must be used somewhere in \texttt{LB42} other than the first column.\footnote{If it is not used at all, or is only used in the first column, then the proof against Case 2 still goes through.}  If all Carry tile east glues are 0h or 1h, the west glue of W2 must be 0h or 1h.  Then a W2 \# south glue would clash with either B1 or W1, so instead it must be different.

Then W2 can only appear in the first row.  Its distinct south glue means it can not be placed above the \# north glue of B1 and W1, and it cannot self-stack, or else its distinct north glue at $(8,2)$ or $(10,2)$ would fail to match the south glues of B1 that appear in those columns in row 3.

Therefore every White tile not in row 1 must be type W1.  From $(3,3)$ and $(4,3)$ we see that W1 stacks horizontally with itself, so its east glue must be equal to its known 1h west glue.  We also previously said that the tile at $(9,4)$ is C2.  We now know that, like its 1h east glue, it has west glue 1h from W1 at $(8,4)$.

At this point, assume C3 has east glue 1h.  Then no Carry tile can be placed at $(7,3)$.  C1 has west glue 0h which fails to match the east glue of W1 at $(6,3)$; C2 and C3 have east glue 1h which fail to match the west glue of B1 at $(8,3)$.

Alternatively assume C3 has east glue 0h.  Then no Carry tile can be places at $(7,2)$.  C1 and C3 have east glue 0h which fails to match the west glue of W1 at $(8,2)$.  C2 has west glue 1h which fails to match the east glue of B1 at $(6,2)$.

So it must be that the east glue of C3 is distinctly 2h.  C3 must be used or the proof against Case 4 goes through.  So with B1 and W1 west glues already assigned, it must be W2 that has west glue 2h.

We also realize that it must be C1 placed to the left of every Black tile, as the three Carry tile east glues are all distinct.  Then we know C1 self-stacks vertically at $(2,1)$ and $(2,2)$, so we can say it has north and south glue g.  We also see that the tile at $(9,4)$, previously decided to be C2, has C1 to its north and south.  Therefore its north and south glues must also be g.

Now we must place C3 at $(7,2)$.  C1 cannot be placed there because its 0h east glue cannot match the 1h or 2h west glues of the White types.  C2 cannot be placed there because then it would have a 0h west glue, and its (south,west) glue pair of (g,0h) would clash with C1.

Then what is the south glue of C3?  By placing C3 at $(7,2)$, it has west glue 0h, so it cannot also have south glue g, or it would clash with C1.  So it has distinct south glue g2.  Neither C1 nor C2 has north glue g2, so C3 must also be placed at $(7,1)$, implying that C3 self-stacks and its north glue is also g2.  But then no tile can be placed at $(7,3)$.  The g south glue of C1 and C2 does not match the north glue of C3; and the 2h east glue of C3 does not match the 0h west glue of B1 at $(8,3)$, and we have reached our contradiction.\\
\textbf{Case 6.  3 Carry types, 2 Black, 1 White.}  This case proceeds similarly to Case 5.

This time we start the same as Case 3.  There is 1 White type W1 with glues (n,s,e,w) equal to (\#,\#,1h,1h).  We will say C2 is the type at $(9,4)$ and has east and west glues 1h.  From $(10,3)$, we can say that type B1 has north and south glue \#, and then must have west glue distinct from 1h, so call it 0h.  The tile at $(9,3)$ has east glue 0h, so we can say that it is type C1.

Now we try to rule out the possibility that C3 has east glue 0h or 1h.  By contradiction, assume that it does.  Then every Carry tile has east glue 0h or 1h.  As a result, B2 must have a south glue distinct from \#.  To see this, we realize that B2 must be used somewhere in \texttt{LB42} other than the first column.  If all Carry tile east glues are 0h or 1h, the west glue of B2 must be 0h or 1h.  Then a B2 \# south glue would clash with either B1 or W1, so instead it must be different.

Then B2 can only appear in the first row.  Its distinct south glue means it can not be placed above the \# north glue of B1 and W1, and it cannot self-stack, or else the distinct north glue at $(3,2), (4,2),$ or $(6,2)$ would fail to match the south glues of W1 that appear in those columns in row 3.

Therefore every Black tile not in row 1 must be type B1.  From $(3,2)$ and $(4,2)$ we see that B1 stacks horizontally with itself, so its east glue must be equal to its known 0h west glue.  We also previously said that the tile at $(9,3)$ is C1.  We now know that, like its 0h east glue, it has west glue 0h from B1 at $(8,3)$.

At this point, assume C3 has east glue 1h.  Then no Carry tile can be placed at $(7,3)$.  C1 has west glue 0h which fails to match the east glue of W1 at $(6,3)$; C2 and C3 have east glue 1h which fail to match the west glue of B1 at $(8,3)$.

Alternatively assume C3 has east glue 0h.  Then no Carry tile can be places at $(7,2)$.  C1 and C3 have east glue 0h which fails to match the west glue of W1 at $(8,2)$.  C2 has west glue 1h which fails to match the east glue of B1 at $(6,2)$.

So the east glue of C3 is distinctly 2h.  C3 must be used or the proof against Case 3 goes through.  So with B1 and W1 west glues already assigned, it must be B2 that has west glue 2h.

We also realize that it must be C2 placed to the left of every White tile, as the three Carry tile east glues are all distinct.  Then we know C2 self-stacks vertically at $(2,3)$ and $(2,4)$, so we can say it has north and south glue g.  We also see that the tile at $(9,3)$, previously decided to be C1, has C2 to its north and south.  Therefore its north and south glues must also be g.

Now we must place C3 at $(7,3)$.  C2 cannot be placed there because its 1h east glue cannot match the 0h or 2h west glues of the Black tiles.  C1 cannot be placed there because then it would have a 1h west glue, and its (south,west) glue pair of (g,1h) would clash with C2.

But then what is the south glue of C3?  By placing C3 at $(7,3)$, it has west glue 1h, so it cannot also have south glue g, by uniqueness.  So it has distinct south glue g2.  Neither C1 nor C2 has north glue g2, so C3 must also be placed at $(7,2)$.  But then its east glue of 2h clashes with the 1h west glue needed by $(8,2)$, and we have reached a final contradiction.  $\bullet$
\end{proof}
With \texttt{LB42} and Lemma~\ref{thm:lb42}, we have shown that if we insist on using less than 4 tile types colored Carry Blue, we must pay a price by adding at least two tile types colored Black or White.  We now leverage this result by adding copies of the colors Black and White, and copies of \texttt{LB42} as sub-patterns \texttt{LB43} through \texttt{LB46}.

Of 46 tile types, we have assigned 34; 1 more is known to be UAB or Black; another 1 more is known to be UAW or White; 10 are free to be any color.
\begin{figure}[tb]
\begin{center}
\includegraphics[width=420pt]{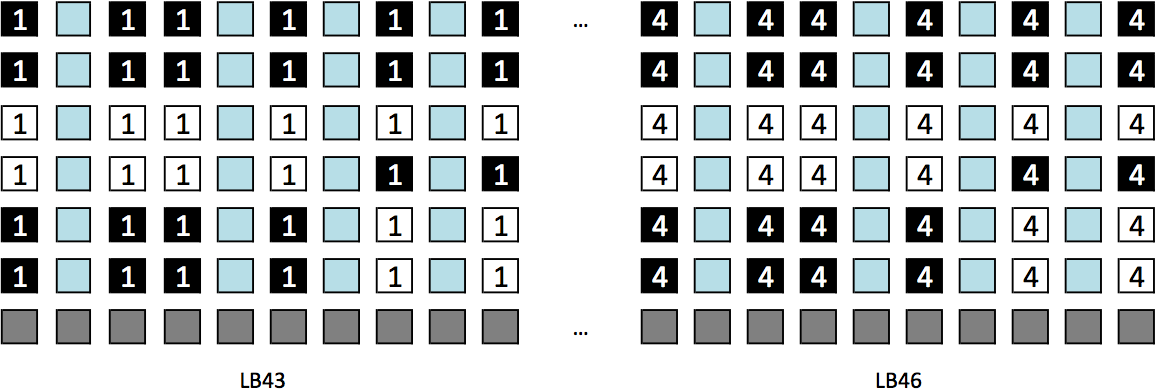}
\end{center}
\caption{Black1 through Black4 and White1 through White4 are copies of the original Black and White colors.  \texttt{LB43} through \texttt{LB46} are copies of \texttt{LB42} with Black and White replaced by respective copies.  Then copies of Lemma~\ref{thm:lb42} apply to each sub-pattern, which implies that if we insist on a small number of Blue Carry tiles, the number of necessary tiles colored by the Blacks and Whites multiplies to the point of impossibility.}
\end{figure}
\begin{prop}
\label{thm:lb43prop}
If \texttt{LB42} through \texttt{LB46} appear as sub-patterns of \texttt{PATTERN}, along with \texttt{LB1} through \texttt{LB41}, then in any minimum tile type set solving \texttt{PATTERN}, either there must be at least 4 Carry tile types colored Blue as we suggest in $T$; or the black box must find a minimum tile type set larger than 46.
\end{prop}
\begin{proof}
We assume that the black box can find a solution using at most 46 tile types.  Now by contradiction, assume that it can do so using fewer than 4 Carry Blue types.

Then from Lemma~\ref{thm:lb42}, we know that it must use at least 4 total tile types colored Black and White.  But by symmetry, Lemma~\ref{thm:lb42} would apply to each pair of colors Black1 and White1 through Black4 and White4 as well, because we include auxiliary patterns \texttt{LB43} through \texttt{LB46}.  Then we must also use at least 4 total types of each of the 4 pairs of Black and White copies.

Between the original sub-pattern \texttt{LB42} and its 4 copies, this will require us to use at least 10 additional free tile types, two of which can be the types that are ``unknown but restricted" to two colors.  This leaves 2 types free.

Now we return to \texttt{LB6} through \texttt{LB41} in Fig.~\ref{fig:lb6} and Fig.~\ref{fig:lb24}.  RAB and RAW already have lower bounds of 2 types each.  With 2 free types remaining available to us, we see that either RAB or RAW must have a maximum of 3 types assigned its color.

But with at most 2 free types that can be assigned to the Purple colors P1, ..., P8, the conditions are in place to apply Lemma~\ref{thm:lb6b}, which would require us to have at least 4 types colored RAB and at least 4 more types colored RAW.  This contradiction shows that we must use at least 4 Carry Blue types, or increase the size of the tile type set.  $\bullet$
\end{proof}
We have reached the following state: of 46 tile types, we have assigned 36; 1 more is known to be UAB or Black; another 1 more is known to be UAW or White; 8 are free to be any color.  With only 8 free types left to take the colors Black or P1, ..., P8, we are in a position to prove the lower bounds of 4 for both RAB and RAW.  To do so we will also need sub-patterns \texttt{LB47} through \texttt{LB82} in Fig.~\ref{fig:lb47}.
\begin{figure}[tb]
\begin{center}
\includegraphics[width=350pt]{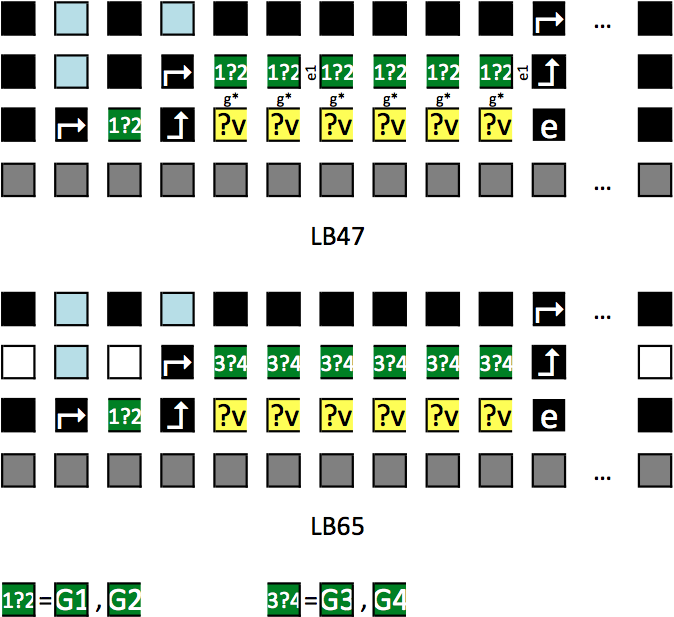}
\end{center}
\caption{There are 18 copies of each sub-pattern shown, for each choice of two Green colors and the 9 colors represented by variable-Yellow.  \texttt{LB47} through \texttt{LB82} are first used in Proposition~\ref{thm:lb43} to show that if Black and Purple colors use all remaining free types, then a unique UAB type leads to a contradiction.  Later they are used again to show that we need fully 4 types of color UAB in Proposition~\ref{thm:lb47ua}.}
\label{fig:lb47}
\end{figure}
\begin{prop}
\label{thm:lb43}
If \texttt{LB1} through \texttt{LB82} appear as sub-patterns of \texttt{PATTERN}, then in any minimum tile type set solving \texttt{PATTERN}, either there must be at least 4 tile types colored RAB with distinct east glues, and at least 4 tile types colored RAW with distinct east glues, as we suggest in $T$; or the black box must find a minimum tile type set larger than 46.
\end{prop}
\begin{proof}
By symmetry we only prove the case for RAB.  We assume that the black box can find a solution using at most 46 tile types.  The current lower bound for every color Black, P1, ..., P8 is 1 tile.  There are 9 total tile types available to be assigned these colors- one is restricted to be Black or UAB, and 8 more free.

Therefore if we do not assign all 9 of these tile types to be second types of colors Black, P1, ..., P8, then Lemma~\ref{thm:lb6b} would be immediately applicable to give us the desired result.  So assume these previously unassigned types are colored as necessary to maintain a contradiction of the stated result.

Now there are no more free tile types, except one type restricted to be White or UAW.  Therefore there is only one tile type of each of the colors UAB, G1, G2, G3, and G4.  Consider the unique tile type colored UAB and call its west glue e1 as labeled in Fig.~\ref{fig:lb47}.

By the uniqueness of types of the listed colors, this must be the east glue of every type colored Green.  We also see that for colors G1 through G4, the unique tile type of each color self-stacks horizontally in \texttt{LB47} through \texttt{LB82}.  Therefore each of these tile types also has west glue e1.

If all four tile types colored G1 through G4 have west glue e1, then uniqueness requires that they have distinct south glues.  This would in turn require at least four tile types of each Purple color, which is impossible- there are theoretically two tile types of each Purple color and we have zero free types remaining.  Having reached a contradiction, we can apply Lemma~\ref{thm:lb6b} directly.  $\bullet$
\end{proof}
Proposition~\ref{thm:lb47ua} proves the lower bound of 4 for each of the colors UAB and UAW, the last lower bound proofs needed for Section 4.  It uses \texttt{LB47} through \texttt{LB82} in Fig.~\ref{fig:lb47} again for UAB, and \texttt{LB83} through \texttt{LB118} shown in Fig.~\ref{fig:lb83} and Fig.~\ref{fig:lb101} for the symmetrical argument for UAW.

It should be understood that there are 18 copies of \texttt{LB47}, one for each combination of: `1?2' equal to G1 or G2, and variable-Yellow assigned each of its 9 possible colors.  For instance, we say that \texttt{LB47} is the copy with color G1 appearing in indexes $(5,2)$ through $(10,2)$, and color Black in indexes $(5,1)$ through $(10,1)$.  And similarly there are 18 copies of \texttt{LB65}.  Then there are 9 copies of each UAW sub-pattern \texttt{LB83}, \texttt{LB92}, \texttt{LB101}, and \texttt{LB110}.

Of 46 tiles, we have assigned 40; 1 more is known to be UAB or Black; another 1 more is known to be UAW or White; 4 are free to be any color.
\begin{prop}
\label{thm:lb47ua}
If \texttt{LB47} through \texttt{LB118} appear as sub-patterns of \texttt{PATTERN}, along with \texttt{LB1} through \texttt{LB46}, then in any minimum tile type set solving \texttt{PATTERN}, either there must be at least 4 tile types colored UAB with distinct west glues and at least 4 tile types colored UAW with distinct west glues as we suggest in $T$; or the black box must find a minimum tile type set larger than 46.
\end{prop}
\begin{figure}[tb]
\begin{center}
\includegraphics[width=330pt]{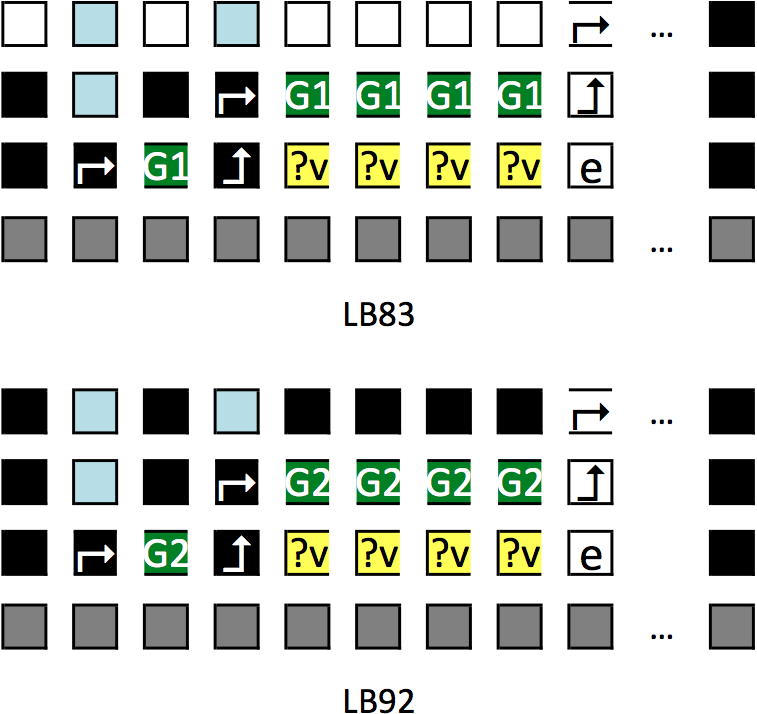}
\end{center}
\caption{\texttt{LB83} through \texttt{LB100} present half of the sub-patterns for UAW that are analogous to those for UAB in Fig.~\ref{fig:lb47}.  The other half are in Fig.~\ref{fig:lb101}.  The sub-patterns are presented 18 at a time for UAB and 9 at a time for UAW because the roles of the different tile types in subtraction require more arrangements of the colors Black and White in the first, third, and last columns across the various pattern designs for UAW than for UAB.}
\label{fig:lb83}
\end{figure}
\begin{figure}[tb]
\begin{center}
\includegraphics[width=330pt]{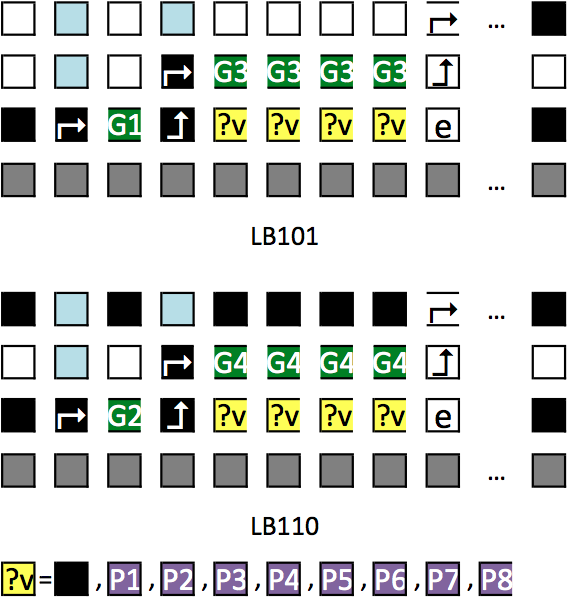}
\end{center}
\caption{The other half of the UAW patterns, referenced in Fig.~\ref{fig:lb83}.  The colors represented by variable-Yellow are shown here reflect all 72 sub-patterns including those in Fig.~\ref{fig:lb47} and Fig.~\ref{fig:lb83}.}
\label{fig:lb101}
\end{figure}
\begin{proof}
By symmetry we only prove the case for UAB, the UAW proof is the same using the sub-patterns in Fig.~\ref{fig:lb83} and Fig.~\ref{fig:lb101} and the remaining free tile types after the proof of the UAB lower bound.  We assume that the black box can find a solution using at most 46 tile types.

Because we only have 4 free types left that can be any of the 9 variable-Yellow colors, there must be at least one color in the set such that there is only 1 tile type of that color.  Call it $C$, and we analyze the four sub-patterns in which $C$ appears.  Also, to assist clarity in this proof, some glue labels that appear in our analysis are shown in Fig.~\ref{fig:lb47}.

Say the unique tile type colored $C$ has north glue g$^*$.  Now for each Green color G1, G2, G3, G4, there is a sub-pattern where it self-stacks horizontally such that it appears in 6 consecutive indexes, each above a tile colored $C$.  We now proceed similarly to the proof of Lemma~\ref{thm:lb0}.

Consider G1.  Because we only have 4 free tile types left that can take the Green colors, we know that there can be at most 5 types colored R1 with south glue g$^*$, with at most 5 distinct east glues.

Because we have 6 consecutive indexes colored R1 such that each must have a tile with south glue g$^*$ placed in it, it must be the case that the 6th tile colored R1 placed at index $(10,2)$ has an east glue e1 that is a repeat of the east glue of one of the tiles colored R1 somewhere to its left.  WLOG say it is the same as the east glue of $(6,2)$ as shown in Fig.~\ref{fig:lb47}.

We do the same analysis for the three other sub-patterns with $C$ in row 1 and respectively each of R2, R3, and R4 in row 2.  Specifically, the respective east glues e2, e3, e4 of their tiles at $(10,2)$ must each be copies of the east glue of a tile of the same color that is placed farther left.  Again, WLOG, the tiles at $(10,2)$ can all be copies of the tiles at $(6,2)$ of their respective sub-patterns.

Now consider the indexes $(7,2)$ of the 4 sub-patterns with $C$ in row 1.  A tile placed in each sub-pattern at $(7,2)$ must have the same south glue g$^*$, and has a Green color distinct from the others, so they must be different tile types.  Then uniqueness requires that the west glues of the tiles placed at $(7,2)$ in each sub-pattern must be distinct. But the west glues of these tiles are exactly e1, e2, e3, e4, so these 4 glues must be distinct.

Finally, we see that going back to e1, e2, e3, e4 as the east glues of the G1, G2, G3, G4 tiles placed at index $(10,2)$ in each of the 4 sub-patterns, that they must be 4 distinct west glues of tile types colored UAB.  $\bullet$
\end{proof}

Of 46 tile types, we have assigned all 46, allowing Proposition~\ref{thm:section4main}.
\begin{prop}
\label{thm:section4main}
If \texttt{LB1} through \texttt{LB118} appear as sub-patterns of \texttt{PATTERN}, then for any minimum tile type set solving \texttt{PATTERN}, it must have total size at least 46, with the desired lower bounds holding for each color respectively, as we suggest in $T$.
\end{prop}
\begin{proof}
\label{thm:section4final}
Aggregating the results in this Section 4 have established the desired lower bound for each color.  Summing over colors gives the overall lower bound of 46.  $\bullet$
\end{proof}

\section{Glue Interpretations}
In this section, we show that if $T$ has size 46, then the glues of some tile types in $T$ must be isomorphic to the names we give to them in Fig.~\ref{fig:tileset1}.\\
It will be sufficient to only prove the glue scheme for tile types used in \texttt{CIRCUIT}\footnote{Though we do prove glues for the Green tile types to facilitate later results.} because then it is clear that unsolvable \textbf{SS} cannot solve \texttt{CIRCUIT} in 26 tile types.
For convenience, and WLOG, we use the type-naming and glue-naming schemes in $T$.  New patterns are named \texttt{GE\#}, for Glue Enforcement.

We re-use  \texttt{LB42} in Fig.~\ref{fig:lb42} on page~\pageref{fig:lb42} for the next two propositions.
\begin{prop}
\label{thm:bwglues}
If \texttt{LB1} through \texttt{LB118} appear as sub-patterns of \texttt{PATTERN} (and establish a minimum tile type set size of 46),
then in any minimum tile type set solving \texttt{PATTERN}, either the glues of the Black and White tile types are as we suggest in $T$; or the black box must find a minimum tile set larger than 46.\end{prop}
\begin{proof}
We assume that the black box can find a solution using at most 46 tile types.  Then there must be exactly 1 type of each color Black and White, as we exhaust 46 tile types to satisfy all lower bounds established in Section 4.

Now we see in \texttt{LB42} that \texttt{BLACK} self-stacks both vertically and horizontally, therefore we know that it has the same north and south glues, call it \#, and the same east and west glues, call it 0h.\footnote{We use ``h" for horizontal; and soon ``v" for vertical.}  \texttt{WHITE} stacks both above and below \texttt{BLACK}, so its vertical glues are the same.

\texttt{WHITE} also self-stacks horizontally, so it also has its east glue equal to its west glue, call it 1h.  However by uniqueness, the glues 1h and 0h must be distinct, because of the common \# south glues of the tile types.  $\bullet$
\end{proof}

\begin{prop}
\label{thm:blueglues}
If \texttt{LB1} through \texttt{LB118} appear as sub-patterns of \texttt{PATTERN},
then in any minimum tile type set solving \texttt{PATTERN}, either the glues of the Carry Blue tile types are as we suggest in $T$; or the black box must find a minimum tile set larger than 46.
\end{prop}
\begin{proof}
We assume that the black box can find a solution using at most 46 tile types.  Then there must be exactly 4 tile types of color Blue, as we exhaust 46 tile types to satisfy all lower bounds established in the last section.

Now using Proposition \ref{thm:bwglues} we see in \texttt{LB42} that four Blue tile types necessarily have (east,west) glue pairs known to be distinct.  For example, tiles placed at indexes $(2,2), (7,2), (2,3),$ and $(7,3)$ require their glues to be (0h,0h), (0h,1h), (1h,1h), and (1h,0h) respectively.  This exhausts the Blue types available.  Name the distinct types that end up at these indexes \texttt{CA1} through \texttt{CA4} in order.

We notice that \texttt{CA1} stacks with itself at $(2,1)$ and $(2,2)$ (because we have shown that Carry tile types are identifiable by their west-east glue pairs), and \texttt{CA3} stacks with itself at $(2,3)$ and $(2,4)$, and they stack both above and below each other.  Therefore the north and south glues of both types must all be equal, call it 0vc.\footnote{Along with ``v" for vertical, we use ``c" for carry; and soon ``p" for passthrough.}  Note that this glue 0vc must also be present between $(7,3)$ and $(7,4)$, so it is also the north glue of \texttt{CA4}.

Now we look at \texttt{CA2} and notice that it stacks with itself at $(7,1)$ and $(7,2)$.  Therefore its north and south glues must be equal; and additionally they must be the same as the south glue of \texttt{CA4}.  However, they cannot be 0vc, otherwise the \texttt{CA2} tile attaching at $(7,2)$ would not do so uniquely (it would clash with \texttt{CA1} at $(3,2)$).  So this second group of glue assignments must be distinct, call it 1vc.  Neither 0vc nor 1vc can be the same as \#, by uniqueness.  $\bullet$
\end{proof}
Propositions~\ref{thm:bwglues} and~\ref{thm:blueglues} have established that all tile types of colors Black, White, and Carry Blue must have glue names matching those we assign in $T$ (gnoring isomorphisms) for all tile sets able to self-assemble \texttt{CIRCUIT}.  The next step is to obtain the same result for el-Black and el-White.
\begin{prop}
\label{thm:elglues}
If \texttt{LB1} through \texttt{LB118} appear as sub-patterns of \texttt{PATTERN},
then in any minimum tile type set solving \texttt{PATTERN}, either the glues of the el-Black and el-White tile types are as we suggest in $T$; or the black box must find a minimum tile type set larger than 46.
\end{prop}
\begin{proof}
We assume that the black box can find a solution using at most 46 tile types.  Then there must be exactly 2 tile types of each color el-Black and el-White, as we need 46 types to satisfy all lower bounds established by Section 4.  We name the types as in $T$.  We only prove the el-Black case, the el-White case is symmetric.

Now we see in \texttt{LB1} (Fig.~\ref{fig:lb1}, page~\pageref{fig:lb1}) that the two \texttt{el-BLACK} types necessarily have known (east,west) glue pairs equal to (0h,0h) and (1h,1h) from the known horizontal glues of Black and White neighbors.  Further, \texttt{el-BLACK1} tiles are placed both above and below a tile colored \texttt{el-BLACK0}, which self-stacks vertically, so all of their vertical glues must be the same, call it 0vp.

In the symmetrical case, call the vertical glue of the el-White types 1vp.  0vp and 1vp must be distinct from \#, 0vc, and 1vc, by uniqueness.  $\bullet$
\end{proof}
The last goal is to prove that the glue schemes of the RA and UA tiles must match the names in $T$.  The next two results make it much more straightforward.  Now that we have Proposition~\ref{thm:elglues} proving the glues for the el-Black and el-White tile types, we can prove the following Lemma describing their behavior as horizontal ``pass-throughs."  
\begin{lem}
\label{thm:ellipsis}
If \texttt{LB1} through \texttt{LB118} appear as sub-patterns of \texttt{PATTERN},
then in any minimum tile type set solving \texttt{PATTERN}, either it is true that for any number $k\geq1$ of tiles colored el-Black and el-White placed consecutively in a horizontal row, the east glue of the eastern most tile and the west glue of the western most tile must be the same; or the black box must find a minimum tile type set larger than 46.
\end{lem}
\begin{proof}
We assume that the black box can find a solution using at most 46 tile types.  Then by Proposition~\ref{thm:elglues}, the tiles colored el-Black and el-White are correctly described in $T$.  A quick check shows us that every tile of these colors ``preserves" glues horizontally, i.e., has the same east and west glues.  A trivial induction gives the result.  $\bullet$
\end{proof}
We can now claim that horizontal glues are the same on both sides of an ellipsis in a sub-pattern if it is implied that all hidden tiles are el-Black or el-White.  Proving the glues of the Green tiles is the next supporting step.
\begin{prop}
\label{thm:greenglues}
If \texttt{LB1} through \texttt{LB118} appear as sub-patterns of \texttt{PATTERN},
then in any minimum tile type set solving \texttt{PATTERN}, either the glues of the tile types of Green colors G1, G2, G3, G4 are as we suggest in $T$; or the black box must find a minimum tile type set larger than 46.
\end{prop}
\begin{proof}
We assume that the black box can find a solution using at most 46 tile types.  Then there must be exactly 1 type of each of the four Green colors G1, G2, G3, G4, as we need 46 types to satisfy all lower bounds.  We name the types as in $T$.

Consider \texttt{LB6} and \texttt{LB15} in Fig.~\ref{fig:lb6}.  Because Black is one of the variable-Yellow colors, we can see that the unique tile type of each Green color stacks both above and below Black's unique tile type and its \# north and south vertical glues.  Therefore every Green type must have north and south glue \#.

Then in \texttt{LB47} and \texttt{LB65} in Fig.~\ref{fig:lb47}, we can see that each of the four types stacks horizontally with itself, so each has east glue equal to west glue.  However, we also know from Proposition 14 that their west glues are all distinct.  Uniqueness also requires that they are all distinct from 0h and 1h, the only other two horizontal glues we have identified so far.  Therefore we can give to them the distinct names as in $T$.  $\bullet$
\end{proof}
With Lemma~\ref{thm:ellipsis} and Proposition~\ref{thm:greenglues} in place, we can get one more result using the \texttt{LB\#} patterns before we need new sub-patterns to finish.  Now Proposition~\ref{thm:threearrowglues} should be intuitive here because we can look at the RA and UA tiles in the second row of the necessary sub-patterns and see that recent results tell us the exact tiles placed at three of the neighboring indexes.
\begin{figure}[tb]
\begin{center}
\includegraphics[width=420pt]{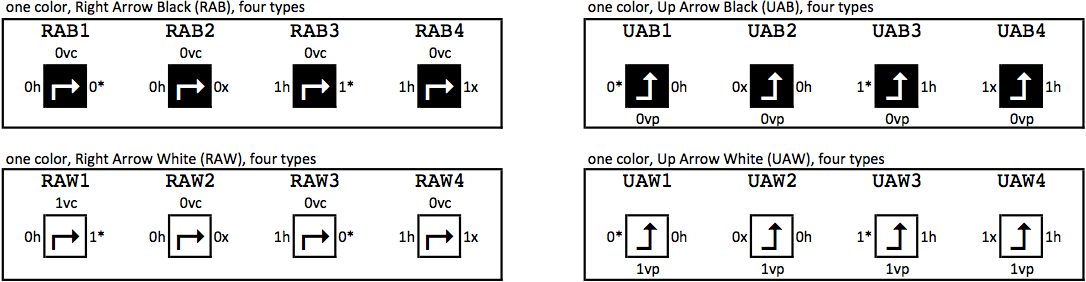}
\end{center}
\caption{After Proposition~\ref{thm:threearrowglues}, the glue names appearing above are necessary for tile type sets $T$ of size 46 that solve our reduction \texttt{PATTERN}.}
\label{fig:recaparrows}
\end{figure}
\begin{prop}
\label{thm:threearrowglues}
If \texttt{LB1} through \texttt{LB118} appear as sub-patterns of \texttt{PATTERN},
then in any minimum tile type set solving \texttt{PATTERN}, either the (south,east,west) glues of all UAB and UAW tile types, and the (north,east,west) glues of all RAB and RAW tile types, are as we suggest in $T$; or the black box must find a minimum tile type set larger than 46.
\end{prop}
\begin{proof}
We assume that the black box can find a solution using at most 46 tile types.  Then there must be exactly 4 tile types of each of the colors UAB, UAW, RAB, and RAW, as we need 46 types to satisfy all lower bounds.  We name the tiles as in $T$.  We only prove the cases for the UAB and RAB tiles, the others are symmetric in patterns \texttt{LB83} through \texttt{LB118} in Fig.~\ref{fig:lb83} and Fig.~\ref{fig:lb101}.

Consider the UAB tile at $(11,2)$ in \texttt{LB47} in Fig.~\ref{fig:lb47} on page~\pageref{fig:lb47}.\footnote{The corresponding index in \texttt{LB83} will be $(9,2)$ for the UAW tiles.}  Its east glue must be 0h by Lemma~\ref{thm:ellipsis}.  To its left is the unique G1 tile type with known east glue 0*, which becomes the west glue of this UAB tile.  Both el-Black types have the same 0vp north glue, so the UAB south glue is 0vp, regardless of what el-Black tile is placed at $(11,1)$.  The same direct analysis lets us solve for the (south,east,west) glues of the three other UAB types in the other three sub-patterns with their respective, distinct Green tiles at $(10,2)$.

Now again in \texttt{LB47}, consider the RAB type at $(4,2)$.  Its west glue must be 0h.  From its east G1 neighbor, its east glue is 0*.  And because Carry Blue types can be determined by their (east,west) glue pairs, we know that this RAB tile's northern neighbor at $(4,3)$ is of type \texttt{Ca1} with its 0vc south glue.  The same direct analysis lets us solve the (north,east,west) glues of the three other RAB types in the other three patterns with their respective, distinct Green tiles at $(4,2)$.  $\bullet$
\end{proof}

Fig.~\ref{fig:recaparrows} gives a summary of the glue names just confirmed by Proposition~\ref{thm:threearrowglues}.  We need 12 more sub-patterns to prove the glues missing in Fig.~\ref{fig:recaparrows}.

\texttt{GE1} through \texttt{GE6} presented in Fig.~\ref{fig:ge1} will be used to force the last glue names of the RAB and UAB tile types.  And symmetrically once again, \texttt{GE7} through \texttt{GE12} in Fig.~\ref{fig:ge7} will serve the same purpose for RAW and UAW.\\
It should be understood  that \texttt{GE4} is then a copy of \texttt{GE1}, except that it has G1 replaced by G2, and G3 replaced by G4.  Similarly, \texttt{GE5} copies \texttt{GE2}.
\begin{figure}[tb]
\begin{center}
\includegraphics[width=350pt]{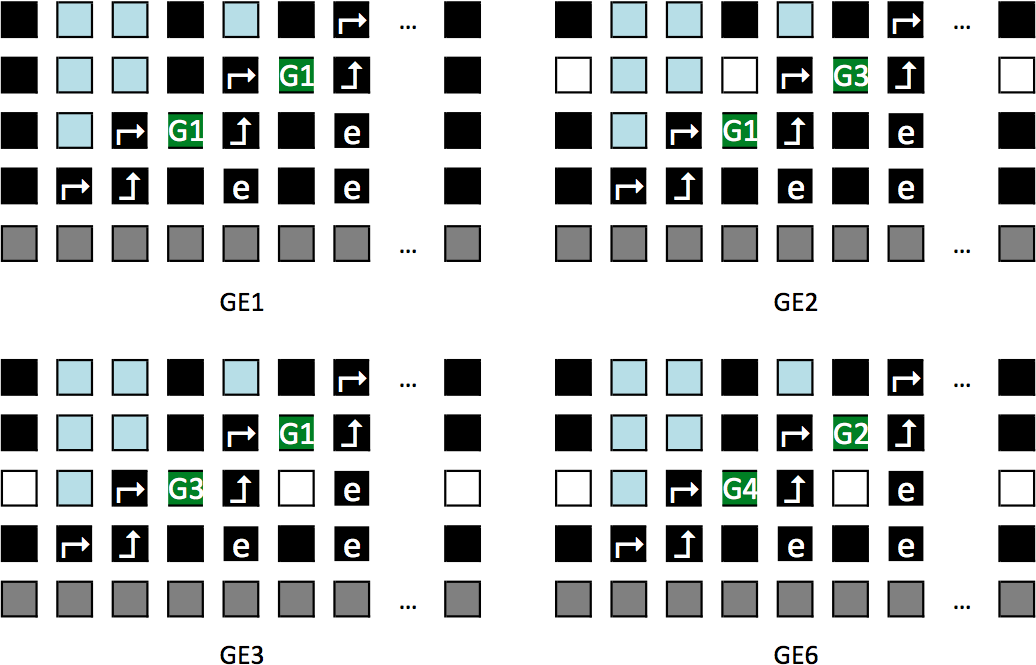}
\end{center}
\caption{Green colors G1 and G3 are associated with subtraction ON signal.  In, sub-patterns \texttt{GE1}, \texttt{GE2} and \texttt{GE3}, we can identify RAB and UAB types by their neighbors, and force the vertical glue name 0* as needed.  Then we use three more patterns replacing G1 and G3 with G2 and G4 respectively to force 0x (with \texttt{GE6} a ``copy" of \texttt{GE3}).}
\label{fig:ge1}
\end{figure}
\begin{prop}
\label{thm:lastarrowglues}
If \texttt{GE1} through \texttt{GE12} appear as sub-patterns of \texttt{PATTERN}, along with \texttt{LB1} through \texttt{LB118},
then in any minimum tile type set solving \texttt{PATTERN}, either
\begin{itemize}
\item the south glues of \texttt{RAB1} and \texttt{RAB3} and the north glues of \texttt{UAB1} and \texttt{UAB3} are all the same, and
\item the south glues of \texttt{RAB2} and \texttt{RAB4} and the north glues of \texttt{UAB2} and \texttt{UAB4} are all the same, and
\item the south glues of \texttt{RAW1} and \texttt{RAW3} and the north glues of \texttt{UAW1} and \texttt{UAW3} are all the same, and
\item the south glues of \texttt{RAW2} and \texttt{RAW4} and the north glues of \texttt{UAW2} and \texttt{UAW4} are all the same,
\end{itemize}
as we suggest in $T$; or the black box must find a minimum tile type set larger than 46.
\end{prop}
\begin{figure}[tb]
\begin{center}
\includegraphics[width=350pt]{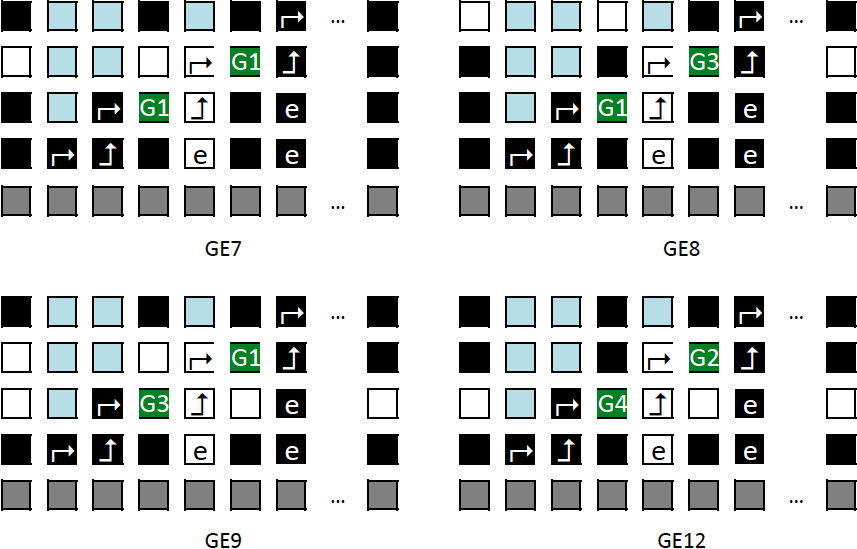}
\end{center}
\caption{\texttt{GE7} through \texttt{GE12} are the analogues to Fig.~\ref{fig:ge1} for RAW and UAW.}
\label{fig:ge7}
\end{figure}
\begin{proof}
By symmetry with all three other cases, we only prove the case for the 0* glues (the first claim).  We assume that the black box can find a solution using at most 46 tile types, and we start with the glue scheme as in Fig.~\ref{fig:recaparrows}.

Now there must be some north glue for tile type \texttt{UAB1}, by choice call it 0*.  We know a tile of type \texttt{UAB1} must be placed at $(5,2)$ of \texttt{GE1} in Fig.~\ref{fig:ge1} because the south and west glues match.  Then we also know that 0* can't be the same as any of the previously existing vertical glues.  If it was, then we wouldn't be able to place a tile colored RAB at $(5,3)$, as its (south,west = 0h) glue pair would require that a tile of a previously described color attach at this index.

In fact, we know exactly what tile type is placed at $(5,3)$- it is \texttt{RAB1} because of its east glue.  So we see that the south glue of \texttt{RAB1} is also 0*.  If we look at sub-pattern \texttt{GE2}, we see that a tile of type \texttt{UAB1} is again placed at $(5,2)$, this time with type \texttt{RAB3} above it, so the south glue of \texttt{RAB3} is 0*.  Finally we consider sub-pattern \texttt{GE3}.  This time it is type \texttt{UAB3} that must be placed at $(5,2)$, and match its north glue to the 0* south glue of \texttt{RAB1}.

We also note that 0*, 0x, 1*, and 1x must all be distinct, or tiles clash in all sub-patterns at $(5,3)$.  $\bullet$
\end{proof}
\begin{prop}
\label{thm:section5main}
If \texttt{GE1} through \texttt{GE12} appear as sub-patterns of \texttt{PATTERN}, along with \texttt{LB1} through \texttt{LB118},
 then in any minimum tile type set solving \texttt{PATTERN}, the glues of all tile types that can be used in \texttt{CIRCUIT} must be the same as we suggest in $T$, up to isomorphic symmetry; or the black box must find a minimum tile type set larger than 46.
\end{prop}
\begin{proof}
\label{thm:section5final}
Aggregating the results in this Section 5 correctly forces all glue names to match those assigned in $T$, up to isomorphic symmetry.  $\bullet$
\end{proof}
\section{Integration and Correctness}
In this Section we prove that 29-PATS is NP-hard.  We start by mentioning and dismissing a minor concern with the design- the possibility that the black box could choose a set of elements to subtract such that at some point, the running total ``becomes negative" after subtracting some element; then it could keep subtracting elements and eventually reach 0, though this subset does not actually sum to the original target $n$.

Our design precludes this possibility by requiring the height of the pattern, equal to $\max(\lceil\log_2(n+\Sigma_{\text{all}})\rceil,21)$, to be sufficiently large: if the running total ever ``goes negative," then even by subtracting every remaining element of $S$, we cannot possibly make it back down to 0 again.

When an element is subtracted from the running total, the operation that we are actually performing is subtraction-mod-$2^{\text{row\#}}$, and we use this idea to join together sub-patterns.  
We can verify that every color sub-pattern begins and ends with a column of squares colored by Black and White, or their copies.  Then \textit{every east glue and every west glue of every sub-pattern is either 0h or 1h}.  Taking the italicized text as an invariant gives us the needed result.
\begin{prop}
\label{tim:splice_together}
Given 2 sub-patterns $P_1$ and $P_2$ of the same height $\text{row\#}$ such that all of the east glues of $P_1$ and all of the west glues of $P_2$ are 0h or 1h, we can join them together into one sub-pattern $P^*$ using the tiles of $T$.
\end{prop}
\begin{proof}
As just explained, our sub-pattern that represents the subtraction of one element from a number ``written" vertically is actually subtraction-mod-$2^{\text{row\#}}$.  Then we are free to think of both the east glues of $P_1$ and the west glues of $P_2$ as respectively input and output numbers mod $2^{\text{row\#}}$, and we can design the needed subtraction-mod-$2^{\text{row\#}}$ to take place in between them.  $\bullet$
\end{proof}
We have a full description of \texttt{PATTERN}.  It lists \texttt{CIRCUIT}, \texttt{LB1} through \texttt{LB118}, and \texttt{GE1} through \texttt{GE12} in order and then joining successively listed sub-patterns together as described in Proposition~\ref{tim:splice_together}.  
Theorem~\ref{thm:equivalence} combines the goals given at the beginning of Section 3.
\begin{thm}
\label{thm:equivalence}
For any instance of Subset Sum \textbf{SS}, if we reduce it to the 29-color pattern \texttt{PATTERN} and submit it to our black box as input, the size of a minimum tile type set that we receive as output will be equal to 46 if and only if \textbf{SS} is satisfiable.
\end{thm}

\begin{proof}
From Proposition~\ref{thm:section4main} we know that the minimum tile set must have size at least 46; and from Proposition~\ref{thm:section5main} we know that if a minimum tile set has size 46, then it must be isomorphic to $T$, i.e., $T$ describes the \textit{necessary} tile set if the minimum tile set has size 46.

From Proposition~\ref{thm:section3main}, we know that $T$ is sufficient for any solvable \textbf{SS}.  Necessity and sufficiency of $T$ for solvable \textbf{SS} completes the ``if" direction of our result.  Now we show that $T$ is not sufficient if \textbf{SS} is not solvable.

We start by analyzing the extent of the ability of the black box to actually make decisions about glue choice for the sub-pattern \texttt{CIRCUIT}, if it intends to use only 46 tiles.  From our discussion in Section 3, we know that 
it is only for the southern glues of RAB and RAW tiles in the first row that the black box actually gets to make a choice.  For RAB tiles it can choose 0* or 0x; for RAW it can choose 1* or 1x.  Again by Proposition~\ref{thm:section3main} these choices exactly correspond to subtracting an element of $S$ or not.

However, by the assumption that we are working with an instance \textbf{SS} that is unsolvable, there is no sequence of seed glue choices for RAB and RAW tiles such that the east glues of the second-to-last column in \texttt{CIRCUIT} are all 0h, to match the uniformly 0h west glues required by the last column's uniformly Black squares.

Then because $T$ is forced when we are restricted to at most 46 tile types, it must be the case that for inputs \texttt{PATTERN} that result from reducing unsolvable instances of \textbf{SS}, the black box finds a minimum tile set of size strictly larger than 46.  $\bullet$
\end{proof}
Theorem~\ref{thm:equivalence} directly implies Theorem~\ref{thm:main} (our main result which was stated in the Introduction).  An instance of Subset Sum \textbf{SS} is solvable if and only if the black box outputs a minimum RTAS size of 46 for an input of \texttt{PATTERN}, completing the reduction.
\begin{coro}
The original PATS problem cannot be approximated to within a 47/46 ratio.
\end{coro}
\begin{coro}
$k$-PATS is not in PTAS for any $k\geq29$.
\end{coro}

	\bibliographystyle{splncs}

\end{document}